\documentclass[letterpaper, 10 pt, conference]{ieeeconf}  % Comment this line out if you need a4paper

\IEEEoverridecommandlockouts                              % This command is only needed if 
\usepackage{cite}
\usepackage{amsmath,amssymb,amsfonts}
\usepackage{algorithm}
\usepackage{algorithmic}
\usepackage{graphicx}
\usepackage{pgfplots}
\pgfplotsset{width=7cm,compat=1.17}
\usepackage{textcomp}
\usepackage{xcolor}

\usepackage{hyperref}

\newtheorem{theorem}{Theorem}
\newtheorem{corollary}{Corollary}
\newtheorem{assumption}{Assumption}
\newtheorem{remark}{Remark}
\newtheorem{definition}{Definition}

\title{\LARGE \bf 
Robust Learning Model Predictive Control for Periodically Correlated Building Control\\
}
\author{Jicheng Shi$^\dagger$, Yingzhao Lian$^\dagger$, and Colin N. Jones% <this % stops a space
\thanks{$^\dagger$The first two authors contributed equally}
\thanks{This work has received support from the Swiss National Science Foundation under the RISK project (Risk Aware Data-Driven Demand Response), grant number 200021 175627.}
\thanks{Jicheng Shi, Yingzhao Lian and Colin N. Jones are with Automatic Laboratory, Ecole Polytechnique Federale de Lausanne, Switzerland.
        {\tt\small
        $\{$jicheng.shi, yingzhao.lian, colin.jones$\}$@epfl.ch}}%
}

\begin{document}
\maketitle
\thispagestyle{empty}
\pagestyle{empty}

\begin{abstract}
Accounting for more than 40\% of global energy consumption, residential and commercial buildings will be key players in any future green energy systems. To fully exploit their potential while ensuring occupant comfort, a robust control scheme is required to handle various uncertainties, such as external weather and occupant behaviour. However, prominent patterns, especially periodicity, are widely seen in most sources of uncertainty. This paper incorporates this correlated structure into the learning model predictive control framework, in order to learn a global optimal robust control scheme for building operations. 
\end{abstract}

% \begin{IEEEkeywords}
% Learning Model Predictive Control, Robust Model Predictive Control, Periodicity
% \end{IEEEkeywords}

\section{Introduction}
Around 40\% of global energy use comes from residential and commercial buildings~\cite{laustsen2008energy}, which drives  research interest in building control. Maximizing operational efficiency while maintaining occupant comfort is the key objective therein. However, various sources of uncertainty, such as internal heat gain and outdoor temperature, pose significant challenges to building operation. Even though uncertain, most of them reveal prominent patterns, especially periodicity. For example, the campus load is shown to evolve within a periodic envelope in~\cite{ma2011model}. Moreover, the alternation between days and nights endows internal heat gain and external temperature periodic pattern on a daily basis~\cite{gondhalekar2013least}.

Besides uncertainty, most buildings are also operated under a periodic scheme. Such periodicity has been widely adopted in building control applications~\cite{minakais2014groundhog,yan2010iterative}, where iterative learning control (ILC) is the key tool enabling efficient performance refinement~\cite{bristow2006survey}. On the other hand, model predictive control (MPC) is a receding horizon control scheme that optimally computes its control inputs by recurrent forecast into the future.  Its natural integration of optimization objective and constraints populates its applications in building control~\cite{ma2011model, gondhalekar2013least,oldewurtel2012use}. Taking advantages of both ILC and MPC~\cite{lautenschlager2016data},  both control schemes deal with optimality and robustness separately. Instead of splitting the control task, learning model predictive control (LMPC) is an optimization-based control scheme that unifies monotonic performance improvement and safety/robustness~\cite{rosolia2017learning,rosolia2017robust,rosolia2019robust}.

In this work, we incorporate the periodically correlated uncertainty into the LMPC framework, which enables LMPC to handle time-varying dynamics. Moreover, owing to a priori knowledge of periodic correlation, the proposed scheme shows higher data efficiency and lower conservativeness. The detailed contribution of this paper is concluded as follows:
\begin{itemize}
    \item Explore a parametric decomposition scheme to handle correlated  noise.
    \item Propose a novel less conservative robust LMPC scheme for periodically correlated process noise, which is designed for periodic tasks.
    \item Demonstrate the convergence and optimality of the proposed LMPC scheme.
\end{itemize}

In the following, we will first introduce the building control problem and the classic LMPC control law in Section \ref{sect:2}. In Section \ref{sect:3}, we introduce a decomposition approach of the periodically correlated disturbance and the novel LMPC is illustrated. The recursive feasibility and performance guarantee of the proposed LMPC is discussed in \ref{sect:4}. In Section \ref{sect:5} and \ref{sect:6}, we describe how to adapt different initial states and model uncertainty in the proposed framework and validate the proposed scheme with a spring-mass system and a single zone building system.
% An extended version with more elaborated numerical results can be found at \url{https://arxiv.org/abs/2011.13781}.

\textbf{Notation}\\
Set of consecutive integers $\{a,a+1,\dots,b\}$ is denoted by  $\mathbb{N}_a^b$. $\mathbb{A} \ominus \mathbb{B}$ denotes Pontryagin set difference. Let $\eta^j$ denote the value of $\eta$ at $j$th iteration. Given value of $\eta$ at time $t$ as $\eta_t$, its prediction at $k$ is denoted by $\eta_{k|t}$, similarly, we have $\eta_{t|t} := \eta_{t}$. $\{a_i\}_{i=1}^N$ is a countable set of cardinality $N$, whose elements $a_i$ are indexed by $i$. $\vee$ denotes the logic ``or".

\section{Set Up the Stage} \label{sect:2}

\subsection{Problem setting}
In this work, we consider a building operation on a daily basis, where a discrete-time periodic time-varying linear building model~\cite{weber2005optimized} with period $T$,
\begin{align} \label{dynamics}
    x_{t+1} = A_{t}x_t + B_{t} u_t + C_{t} w_t,\; \forall t\;\in\mathbb{N}_0^T\;,
\end{align}
 where states, control inputs and the bounded process noise are denoted by $x \in \mathbb{R}^n$, $u \in \mathbb{R}^m$ and $w \in \mathbb{R}^d$ separately. 
 
This system is manipulated to execute an iterative task, which means at the $j$th iteration, it starts from $x_0^j=x_s $. The states and inputs are required to satisfy the following periodic, convex polytopic  constraints:
\begin{align}
    F_t x_t + G_t u_t \leq f_t, \forall\;t \in\mathbb{N}_0^T\; .
\end{align}

The optimal building operation problem of the $j$th day is as follows:
\begin{subequations}\label{prob_general}
\begin{align}
    J^{j,\ast} = \underset{\{u_{t}^j\}_{t=0}^T}{\min}&\sum_{t=0}^{T}{l_t(x_t^j,u_t^j)} \notag \\
    \text{s.t.} & \;\forall t\in\mathbb{N}_0^T,\;x_0^j = x_s \notag\\
    & x_{t+1}^j = A_t x_t^j + B_t u_t^j + C_tw_t^j ,\label{eqn:prob_general:dyn}\\
    % &v(t) = \pi_t(z_0,v_0,...,v_{t-1},z_t), \\
    &F_t x_t^j + G_t u_t^j \leq f_t , \label{eqn:prob_general:con}
\end{align}
\end{subequations}
where $l_t$ denotes the stage cost at time $t$ and $w_t^j$ represents the unknown disturbance / uncertainty within the $j$th day. The horizon $T$ in Problem~\eqref{prob_general} is in general large in building control. For example, if the control law changes every 10 minutes, then $T$ reaches 144.

\subsection{Learning Model Predictive Control}
 Learning model predictive control (LMPC) is an iterative control scheme proposed to learn infinite/long horizon optimal control trajectories, where a relatively short horizon problem is solved in a moving horizon scheme~\cite{rosolia2017learning}. For the sake of clarity, we elaborate LMPC with a deterministic system (\textit{i.e.} $w=0$ in~\eqref{eqn:prob_general:dyn}). At time $t\in \mathbb{N}_0^T$, the following problem is solved:
\begin{subequations}\label{prob:LMPC_ori}
\begin{align}
    \min\limits_{\{u_{k|t}^j\}_{k=t}^{t+N-1}}&\sum_{k=t}^{t+N-1}{l_k(x_{k|t}^j,u_{k|t}^j)}+Q^j(x_{t+N|t}^j)\label{eqn:LMPC_ori:loss}\\
    \text{s.t.} &\;\forall k \in\mathbb{N}_t^{t+N-1}, \;x_{t|t}^j = x_{t}^j\notag\\
    &x_{k+1|t}^j = A x_{k|t}^j + B_k u_{k|t}^j\notag\\
    &F x_{k|t}^j + G u_{k|t}^j \leq f,\notag \\
    &x_{t+N|t}^j\in\mathbb{SS}^j\label{eqn:LPMC_ori:ss}\;.
\end{align}
\end{subequations}
$Q^j(\cdot)$ in~\eqref{eqn:LMPC_ori:loss} and set $\mathbb{SS}^j$ in~\eqref{eqn:LPMC_ori:ss} are two main components which ensure the safety and monotonic improvement of LMPC. In particular, $\mathbb{SS}^j$ denotes the safe set within which there is at least one control law ensuring system safety. This set is constructed as the convex hull of all observed trajectories before the current iteration $j$. Meanwhile, $Q^j(\cdot)$ is an overestimate of the optimal cost-to-go, which ensures  the cost calculated in~\eqref{eqn:LMPC_ori:loss} overestimates the optimal cost in Problem~\eqref{prob_general}. In particular, $Q^j(\cdot)$ is modelled by  parametric quadratic programming in standard LMPC~\cite{rosolia2017learning2}.

The LMPC control scheme guarantees convergence to the infinite/long horizon solution~\cite{rosolia2020optimality} and has been extended to robust control with additive noise~\cite{rosolia2017robust,rosolia2019robust} and deterministic periodic control~\cite{scianca2020learning}.

\section{Main Results} \label{sect:3}
In this section, the incorporation of correlation information is first introduced by finite order approximation in Section~\ref{sect:noise_decomp}. The adapted LMPC algorithm for the resulting problem is then introduced in Section~\ref{sect:LMPC_corr}.

\subsection{Process Noise Decomposition}\label{sect:noise_decomp}
Most sources of uncertainty in building control reveal significant periodic patterns, such as external temperature and internal heat gain. The main idea behind our approach is to decompose the uncertainty information into periodically correlated and uncorrelated parts (\textit{i.e.} white noise). To proceed, we first make the following assumption.

\begin{assumption}\label{ass:integrable}
    $w_t,\; t\in \mathbb{N}_0^T$ is a bounded stochastic process and $\mathbb{E}(w_t) = a_0,\;\forall \;t\in\mathbb{N}_0^T$.
\end{assumption}

$w_t$ is a stochastic process with finite end time $T$, and $w_t^j$ is a realization of this process. More specifically, if $w_t$ is the process of external temperature, then $w_t^j$ is the temperature trajectory on the $j$th day. Assumption~\ref{ass:integrable} ensures that the process noise on the $j$th day is square integrable with respect to its probability space~\cite{jacod2012probability}. By the Karhunen–Loève theorem~\cite{loeve1977elementary}, $w_t^j$ is decomposed based on Fourier series as
\begin{align} \label{disturbance_inf}
    w_t^j = a_0^j + \sum\limits_{q=1}^\infty a_q^j\sin\left(\frac{2\pi q t}{T}\right)+b_q^j\cos\left(\frac{2\pi q t}{T}\right)\;,
\end{align}
where $a_q = \frac{2}{T}\int_0^Tw(t)\sin(\frac{2\pi q t}{T})dt$ and $b_q$ is defined accordingly. To only preserve the low frequency information, Equation~\eqref{disturbance_inf} is further approximated by
\begin{align}\label{disturbance_finite}
    w_t^j &= a_0^j + \sum\limits_{k=1}^{M} \left[a_q^j\sin\left(\frac{2\pi q t}{T}\right)+b_q^j\cos\left(\frac{2\pi q t}{T}\right)\right]+ w_{r,t}^j\\
    &= a_0^j + w_{\theta^j,t}^j+w_{r,t}^j\;\notag,
\end{align}
where $w_{r,t}^j$ models the truncation error caused by the finite order approximation $w_{\theta^j,t}^j$. In particular, the collection of parameters $\theta^j:=\{a_q^j,b_q^j\}_{q=1}^M$ captures the periodic correlation within the $j$th day, which is bounded as $\theta^j\in\mathbb{W}_\theta,\;\forall\;j$. The residue $w_{r,t}$ is a zero-mean bounded white noise whose variance is $\text{var}(w_{r,t})=\mathbb{E}(\sum_{q=M+1}^\infty(||a_q^j||_2^2+||b_q^j||_2^2))$, which is well defined by Assumption~\ref{ass:integrable} and that preserves the energy of the process noise(\textit{i.e.} Parseval theorem~\cite{stein2011fourier}). To explain~\eqref{disturbance_finite} more specifically, one can consider $w_t$ as the external temperature. In the $j$th day, $a_0^j$ is the averaged temperature, $\{a_q^j,\;b_q^j\}_{q=1}^M$ models the daily evolution of the temperature, while $w_{r,t}^j$ models highly stochastic fast fluctuations. Regarding this interpretation, $a_0^j$ and $\theta^j$ vary among days. Similar to~\eqref{disturbance_inf}, other orthogonal basis functions can be used to approximate specific noise patterns, such as Haar Wavelet basis~\cite{unser2014introduction} for internal heating gains. For the sake of simplicity, we elaborate our method with a simpler model as
\begin{align} \label{disturbance}
    w_t^j &= a_0^j + a_1^j\sin(2\pi t/T) + w_{r,t}^j\;\notag\\
    &= w_{\theta^j,t}^j+w_{r,t}^j, \;\theta^j = \{a_0^j,a_1^j\}\;.
\end{align}

\begin{remark}
 Notice that $\{a_q^j,b_q^j\}_{q=1}^\infty$ are realizations of random variables according to the Karhunen–Loève theorem~\cite{loeve1977elementary}, which means that they are fixed in $w_t^j$. In practice, within each iteration, these parameters can be effectively estimated by different methods, such as Bayesian learning~\cite{robert2007bayesian}. 
 \end{remark}
 
% \begin{remark}
% As one might notice that a stochastic process decomposition is mainly applied to a continuous time stochastic process, the decomposition procedure discussed above implicitly discretizes the process by stochastic integration. In particular, the closedness of the Fourier basis under linear dynamics leads to~\eqref{disturbance}.
% \end{remark}

\subsection{LMPC for correlated noise}\label{sect:LMPC_corr}
As noise are decomposed into a correlated part and an uncorrelated part in~\eqref{disturbance_finite}, they can be handled separately in the robust control problem. In particular, the white noise $w_{r,t}^j$ are handled by standard robust model predictive control methods~\cite{kouvaritakis2016model} (details in Appendix~\ref{app:robust}). The resulting robust form of the long horizon Problem~\eqref{prob_general} is 
\begin{subequations}\label{prob_general_robust}
\begin{align}
    J^{j,\ast} = \underset{\{v_{t}^j\}_{t=0}^T}{\min}&  \sum_{t=0}^{T}{l_t(z_t^j,v_t^j)} \notag \\
    \text{s.t.} & \;\forall t\in\mathbb{N}_0^T,\;z_0^j = x_s \notag\\
    & z_{t+1}^j = A_t z_t^j + B_t v_t^j + C_tw_{\theta^j,t}^j ,\label{eqn:prob_general_robust:dyn}\\
    % &v(t) = \pi_t(z_0,v_0,...,v_{t-1},z_t), \\
    &\bar{F}_t z_t + \bar{G}_t v_t \leq \bar{f}_t, \label{eqn:prob_general_robust:con}
\end{align}
\end{subequations}
where $z_t^j,v_t^j$ denotes the state and input of a nominal system, and~\eqref{eqn:prob_general_robust:con} is the tightened constraint (Appendix~\ref{app:robust}).

Correspondingly, the robust form of the LMPC problem~\eqref{prob:LMPC_ori} is:
% \begin{subequations}
\begin{align}\label{prob_LMPC}
    J_{LMPC}^{j,\ast} &= \min\limits_{\{v_{k|t}^j\}_{k=t}^{t+N-1}}\sum_{k=t}^{t+N-1}{l_k(z_{k|t}^j,v_{k|t}^j)}+Q_{t+N}^{j}(z_{t+N|t}^j)\notag\\
    \text{s.t.} &\;\forall k \in\mathbb{N}_t^{t+N-1}, \;z_{t|t}^j = z_{t}^j\notag\\
    &z_{k+1|t}^j = A_k z_{k|t}^j + B_k v_{k|t}^j + C_kw_{\theta^j, k}^j\notag \\
    &\bar{F}_k z_{k|t}^j + \bar{G}_k v_{k|t}^j \leq \bar{f}_k, \notag\\
    &z_{t+N|t}^j\in\mathbb{SS}_{t+N}^{j}\;.
\end{align}
% \end{subequations}
The daily changing disturbances included in the dynamics and periodic tasks make classic LMPC not applicable, which requires new adaptive algorithms to calculate $\mathbb{SS}^j_t$ and $Q_t^j(\cdot)$. In the following, we show the strategy of constructing these two main components accordingly. To proceed, we first define the following notation for a more compact layout.

At time $t$ of the $j$th iteration, denote by the vectors
\begin{align}
    &\mathbf{v}_t^{j,\ast} = [v_{t|t}^{j,\ast},v_{t+1|t}^{j,\ast},..., v_{t+N-1|t}^{j,\ast}],\\
    &\mathbf{z}_t^{j,\ast} = [z_{t|t}^{j,\ast},z_{t+1|t}^{j,\ast},..., z_{t+N|t}^{j,\ast}].
\end{align}
the optimal input sequence and the resulting state sequence. Then at time $t$, the input applied to the closed-loop system is
\begin{align} \label{clp_input}
    v_t^{j} =\left\{
        \begin{array}{lr}
        v_{t|t}^{j,\ast}, & t+N \leq T,\\
        v_{t|T-N}^{j,\ast}, & t+N > T.
        \end{array}
\right.
\end{align}

In the following, the idea of historical trajectory shifting will enable us to define the adapted safe sets $\mathbb{SS}_t^j$ and $Q$ function $Q_t^{j}(\cdot)$. Consider at a historical $i$th iteration, the vectors
\begin{align}
    \mathbf{z}^i = [z_0^i, z_1^i, ..., z_T^i] \\
    \mathbf{v}^i = [v_0^i, v_1^i, ..., v_T^i] \notag
\end{align}
record the historical  states and inputs in the closed-loop trajectories. When building a safe set for the $j$th iteration, a shifting method is applied on the historical data, $\mathbf{z}^i$ and $\mathbf{v}^i$. For a shifting starting from time step $t$ of the $i$th historical trajectory, denote by $v_{k|t}^{i,j}$ the shifted input, by $z_{k|t}^{i,j}$ the shifted state, 
by  $e_{k|t}^{i,j} = z_{k|t}^{i,j} - z_k^i$ the error state, the shifting follows a procedure:
\begin{align} \label{sys_nom:err_dyn}
    e_{k+1|t}^{i,j} &= \Phi_k e_{k|t}^{i,j} + C_k (w_{\theta^j,k}^j-w_{\theta^i,k}^i) \notag\\
    v_{k|t}^{i,j} &= v_k^i+ K_k e_{k|t}^{i,j} \notag\\
    z_{k|t}^{i,j} &= z_k^{i} + e_{k|t}^{i,j}, \forall k \in\mathbb{N}_t^{T}\notag\\
\end{align}
and $e_{t|t}^{i,j} = 0$, where $K_k$ is chosen to stabilize $\Phi_k=A_k + B_k K_k$. As a result, $z_{k|t}^{i,j}$ and $v_{k|t}^{i,j}$ satisfy the $j$th dynamics:
\begin{align*}
    z_{k+1|t}^{i,j} &= A_k z_{k|t}^{i,j} + B_k v_{k|t}^{i,j} + C_k w_{\theta^i,k}^i +  C_k(w_{\theta^j,k}^j-w_{\theta^i,k}^i)\\
    &= A_k z_{k|t}^{i,j} + B_k v_{k|t}^{i,j} + C_k w_{\theta^j,k}^j
\end{align*}

Note that the shifted states and inputs may result in infeasible shifted data due to constraint violations. The elimination of these infeasible shifted data leads us to the concept of \textit{Feasible Disturbance Set}.

\begin{definition}[Feasible Disturbance Set]
At time $t$ in an historical iteration, the Feasible Disturbance Set $\mathbb{W}_t^{i}$ is defined as:
\begin{align*}
    \mathbb{W}_t^i = \{\theta | &\bar{F}_k(z_k^{i} + e_{k|t}^{i,.}) + \bar{G}_k(v_k^i+ K_k e_{k|t}^{i,.}) \leq \bar{f}_k,e_{t|t}^{i,.} = 0 \\
    &e_{k+1|t}^{i,.} = \Phi_k e_{k|t}^{i,.} + C_k(w_{\theta,k} - w_{\theta^i,k}^i), \forall k \in\mathbb{N}_t^T\}\\
\end{align*}
\end{definition}
After finishing the $j$th iteration and recording the closed-loop states $\mathbf{z}^j$ and inputs $\mathbf{v}^j$, the feasible disturbance set at each time is computed and recorded.

\begin{algorithm} 
  \caption{Safe set}
  \label{alg:ss} 
{
Given historical closed loop states $\mathbf{z}^i$, inputs $\mathbf{v}^i$, \\
feasible disturbance set $\mathbb{W}_t^i$, $\forall t \in \mathbb{N}_0^T, i \in \mathbb{N}_0^{j-1}$ 

\begin{enumerate}
    \item For $i \in \mathbb{N}_0^{j-1}$, $t \in \mathbb{N}_0^{T}$
    \begin{enumerate}
        \item If $\theta^j \in \mathbb{W}_t^i$
        \begin{enumerate}
            \item Compute the shifting from time $t$
                \item[] $[z_{t|t}^{i,j},...,z_{T|t}^{i,j}]$, $[v_{t|t}^{i,j},...,v_{T|t}^{i,j}]$
            \item Add state $z_{k|t}^{i,j}$ to $\mathbb{SS}_k^j$, $\forall k \in \mathbb{N}_t^T$
            \item Compute and record shifted cumulative cost
                \item[] $J_{k|t}^{i,j}(z_{k|t}^{i,j}) = \sum_{r=k}^{T}{l(z_{r|t}^{i,j},v_{r|t}^{i,j})}$ $\forall k \in \mathbb{N}_t^T$
    \end{enumerate}
    \end{enumerate}
\end{enumerate}
}
\end{algorithm}

Now we build the safe set $\mathbb{SS}_t^j$ for the $j$th iteration by Algorithm \ref{alg:ss}. Note in the shifting starting from time $t$, it computes the shifted states from $t$ to $T$ and each shifted state $z_{k|t}^{i,j}$ is added to $\mathbb{SS}_k^j$ correspondingly. Meanwhile, the estimated cost-to-go (\textit{i.e.} $Q_k^j(\cdot)$ in~\eqref{prob_LMPC}) are updated by shifted cumulative costs $J_{k|t}^{i,j}$ as

% The shifted accumulative costs $J_{k|t}^{i,j}$ is also computed and used to build Q function,
\begin{equation}\label{eqn:Qfun}
Q_k^{j}(z) =\left\{
        \begin{array}{lr}
        \underset{(i,t)\in F_k^{j}(z)}{\min}J_{k|t}^{i,j}(z), & \text{if}\; z\in \mathbb{SS}_k^{j}\\
        +\infty, & \text{if} \; z\notin \mathbb{SS}_k^{j}
        \end{array}
\right.
\end{equation}
where $F_k^{j}(z)=\{(i,t): i\in[0,j-1],t\in[0,k]$ with $z_{k|t}^{i,j}=z$, for $z_{k|t}^{i,j} \in \mathbb{SS}_k^{j}\}$. Note different from~\cite{rosolia2017learning}, at the $j$th iteration, $\mathbb{SS}_t^{j}$ and $Q_t^{j}(z)$ are built for each time step $t$.

\begin{remark}
  At each time step $t$ and each shifted state $z$ in $\mathbb{SS}_t$, $Q_t^{j}(z)$ is assigned a value, $J_{k|t^{\ast}}^{i^{\ast},j}$, which is the minimal shifted cumulative cost starting from $z_{k|t^{\ast}}^{i^{\ast},j}=z$. $(i^{\ast},t^{\ast})$  is chosen by the minimizer in~\eqref{eqn:Qfun}:
  \begin{align*}
      (i^{\ast},t^{\ast}) = \underset{(i,t)\in F_k^{j}(z)}{\operatorname{argmin}}J_{k|t}^{i,j}(z),\forall z\in \mathbb{SS}_k^{j}
  \end{align*}
\end{remark}

\begin{assumption} \label{assu:distubance_bound}
     Assume a feasible trajectory at the $0$th iteration, $\{\mathbf{z}^0, \mathbf{v}^0\}$, is given and all the disturbance feasible sets are subject to, $\mathbb{W}_t^0 \supseteq \mathbb{W}_{\theta}$.
\end{assumption}

Assumption \ref{assu:distubance_bound} is standard under the LMPC control scheme. It results in a non-empty safe set $\mathbb{SS}_t^j,\;\forall\;t\in\mathbb{N}_0^T, j\in \mathbb{N}_+$. In practice, Assumption \ref{assu:distubance_bound} is not restrictive as it essentially requires a default feasible control law. It is also noteworthy to point out that neither an historical nor a shifted trajectory are required to achieve a steady state, while this convergence requirement is necessary for classic LMPC. 
% Common robust optimal methods can be used to compute $0$th trajectory with details in \ref{app:shift}.

\begin{remark}
 The online computation increase of the proposed scheme is fair, as feasible disturbance sets $\mathbb{W}_t^{i}$, safe set $\mathbb{SS}_t^j$ and Q function $Q_t^j(\cdot)$ only update at the beginning of each iteration.
\end{remark}

\begin{remark}
Even though this work has a special focus on building control, the proposed scheme can be adopted to most time-varying periodic tasks.
\end{remark}

\section{Properties} \label{sect:4}
In this section, the properties of the proposed LMPC method are presented, including feasibility and performance.
\subsection{Recursive Feasibility}

\begin{theorem}[Recursive Feasibility]
Suppose Assumption 2 is satisfied, then the problem~\eqref{prob_LMPC} is feasible for any time step $t$ at any $j$th iteration.
\end{theorem}

\begin{proof}
% The proof can be found in the \href{https://arxiv.org/abs/2011.13781}{\underline{extended version}}.
By Assumption 2, $0$th iteration offers a shifted trajectory starting from time $0$. Thus, at the $j$th iteration, the shifted state $z_{N|0}^{0,j} \in \mathbb{SS}_N^j$. At the time step $0$ of $j$th iteration, the following shifted state and input vectors is feasible for the problem~\eqref{prob_LMPC}:
\begin{align*}
    &[v_{0|0}^{0,j},v_{1|0}^{0,j},...,v_{N-1|0}^{0,j}],\\
    &[z_{0|0}^{0,j},z_{1|0}^{0,j},...,z_{N|0}^{0,j}].
\end{align*}

Assume at time step $t$ of $j$th iteration, the problem~\eqref{prob_LMPC} is feasible, with the optimal solution $\mathbf{v}_t^{j,\ast}$ and the corresponding state sequence $\mathbf{z}_t^{j,\ast}$. Note that $z_{t+N|t}^{j,\ast} \in \mathbb{SS}_{t+N}^{j}$. By the definition of Q function, $z_{t+N|t}^{j,\ast} = z_{t+N|t^{\ast}}^{i^{\ast},j}$, which is a shifted state at time $t+N$ starting from some time $t^{\ast}\leq t+N$ at $i^{\ast}$th iteration. Then with the corresponding shifted input $
v_{t+N|t^{\ast}}^{i^{\ast},j}$ we have the next shifted state $z_{t+N+1|t^{\ast}}^{i^{\ast},j} = A_{t+N}z_{t+N|t^{\ast}}^{i^{\ast},j} + B_{t+N}v_{t+N|t^{\ast}}^{i^{\ast},j} +C_{t+N}w_{\theta^j,t+N}^j$. From the Algorithm 1, $z_{t+N+1|t^{\ast}}^{i^{\ast},j} \in \mathbb{SS}_{t+N+1}^j$. Thus, time step $t+1$, the input sequence and corresponding state sequence
\begin{align*}
    &[v_{t+1|t}^{j,\ast},v_{t+2|t}^{j,\ast},...,v_{t+N-1|t}^{j,\ast},v_{t+N|t^{\ast}}^{i^{\ast},j}],\\
    &[z_{t+1|t}^{j,\ast},z_{t+2|t}^{j,\ast},...,z_{t+N|t}^{j,\ast},z_{t+N+1|t^{\ast}}^{i^{\ast},j}]
\end{align*}
is feasible. Finally, by induction, the theorem is proved.
\end{proof}

\subsection{Performance}
In this section, we present two results regarding controller performance. At the $j$th iteration, denote the optimal value of the objective function of the problem~\eqref{prob_LMPC} at time step $t$ by $J_{LMPC}^{j,\ast}(z_t^j) = \sum_{k=t}^{N}{l_k(z_{k|t}^{j,\ast},v_{k|t}^{j,\ast})}+Q_{t+N}^{j}(z_{t+N|t}^{j,\ast})$, the closed-loop cumulative cost starting from time $t$ by $J^{j}(z_t^j) = \sum_{k=t}^{T}{l_k(z_{k}^{j},v_{k}^{j})}$.

\begin{assumption} \label{ass:cost_eco}
    Consider a continuous, semi-positive and convex stage cost function $l_t(z,v)\geq0$
\end{assumption}
Different from~\cite{rosolia2017learning}, the stage cost is not limited to a tracking error. An economic cost can be used, like the electricity cost in building control, for example.

\begin{theorem} \label{thm_noninc}
Under Assumption~\ref{ass:cost_eco}, for each $t\in \mathbb{N}_0^{T-N}$ of the $j$th iteration, the cumulative trajectory cost $J^{j}(z_t^j)$ is upper bounded by the shifted trajectory cost $J_{t|t'}^{i,j}(z_{t|t'}^{i,j})$, starting from any $z_{t|t'}^{i,j}=z_t^j \in \mathbb{SS}_t^j$. Specially, if $\theta^j \in \mathbb{W}_0^i$, $J^{j}(z_0^j)\leq J_{0|0}^{i,j}(z_0^j)$.
\end{theorem}

\begin{proof}
% The proof can be found in the \href{https://arxiv.org/abs/2011.13781}{\underline{extended version}}.
At time step $t$($\in \mathbb{N}_0^{T-N})$ of the $j$th iteration, the optimal cost of LMPC is:
\begin{align} \label{pf:nonin}
    J&_{LMPC}^{j,\ast}(z_t^j)  \notag\\ &=\underset{\{v_{k|t}^j\}_{k=t}^{t+N-1}}{min}\sum_{k=t}^{N-1}{l_k(z_{k|t}^j,v_{k|t}^j)}+Q_{t+N}^{j}(z_{t+N|t}^j) \notag\\
    &= l_t(z_{t|t}^{j,\ast},v_{t|t}^{j,\ast}) +\sum_{k=t+1}^{N-1}{l_k(z_{k|t}^{j,\ast},v_{k|t}^{j,\ast})}+Q_{t+N}^{j}(z_{t+N|t}^{j,\ast})\notag\\
    &\geq l_t(z_{t|t}^{j,\ast},v_{t|t}^{j,\ast}) +\sum_{k=t+1}^{N-1}{l_k(z_{k|t}^{j,\ast},v_{k|t}^{j,\ast})}+ \notag\\ & \quad \quad l_{t+N}(z_{t+N|t^{\ast}}^{i^{\ast},j},v_{t+N|t^{\ast}}^{i^{\ast},j}) +Q_{t+N+1}^{j}(z_{t+N+1|t^{\ast}}^{i^{\ast},j})\notag\\ 
    &\geq l_t(z_{t|t}^{j,\ast},v_{t|t}^{j,\ast}) + J_{LMPC}^{j,\ast}(z_{t+1|t}^{j,\ast})
\end{align}
In the first inequity, $z_{t+N|t}^{j,\ast} = z_{t+N|t^{\ast}}^{i^{\ast},j}$, which is a shifted state at time $t+N$ starting from some time $t^{\ast}\leq t+N$ and the inequity comes from the definition of Q function.

Then under Assumption \ref{ass:cost_eco}, from~\eqref{clp_input} and~\eqref{pf:nonin}, $J_{LMPC}^{j,\ast}$ is non-increasing along the closed loop trajectory,
\begin{align} \label{cost_dec}
J_{LMPC}^{j,\ast}(z_{t+1}^{j}) - J_{LMPC}^{j,\ast}(z_t^j)\leq -l_t(z_{t|t}^{j},v_{t|t}^{j}) \leq 0
\end{align}
By~\eqref{clp_input} and~\eqref{cost_dec}, the cumulative trajectory cost $J^{j}(z_t^j)$ is upper bounded by $J_{LMPC}^{j,\ast}(z_t^j)$:
\begin{align}
    J_{LMPC}^{j,\ast}(z_t^j) &\geq l_t(z_{t}^{j},v_{t}^{j})+ J_{LMPC}^{j,\ast}(z_{t+1}^{j}) \notag\\
    &\geq l_t(z_{t}^{j},v_{t}^{j})+l_{t+1}(z_{t+1}^{j},v_{t+1}^{j})+ J_{LMPC}^{j}(z_{t+2}^{j}) \notag\\
    &\geq \sum_{k=t}^{T}{l_k(z_{k}^{j},v_{k}^{j})} = J^{j}(z_t^j)
\end{align}

Then we show the shifted trajectory cost $J_{t|t'}^{i,j}(z_{t|t'}^{i,j})$, starting from any shifted state $z_{t|t'}^{i,j}=z_t^j \in \mathbb{SS}_t^j$, is lower bounded by $J_{LMPC}^{j,\ast}(z_t^j)$:

\begin{align} 
    J&_{t|t'}^{i,j}(z_{t|t'}^{i,j}) = \sum_{k=t}^{T}{l_k(z_{k|t'}^{i,j},v_{k|t'}^{i,j})}\notag\\
    &=\sum_{k=t}^{t+N-1}{l_k(z_{k|t'}^{i,j},v_{k|t'}^{i,j})} + \sum_{k=t+N}^{T}{l_k(z_{k|t'}^{i,j},v_{k|t'}^{i,j})}\notag\\
    &\geq \sum_{k=t}^{t+N-1}{l_k(z_{k|t'}^{i,j},v_{k|t'}^{i,j})} + Q_{t+N}^{j}(z_{t+N|t'}^{i,j})\notag\\
    &\geq \underset{\{v_{k|t}^j\}_{k=t}^{t+N-1}}{min} \sum_{k=t}^{t+N-1}{l_k(z_{k|t}^{j},v_{k|t}^{j})} +Q_{t+N}^{j}(z_{t+N|t}^{i,j})\notag\\
    &=J_{LMPC}^{j,\ast}(z_t^j)
\end{align}
\end{proof}

After execution of the $j$th iteration, if in a new iteration $j'$, the same disturbance parameters occur $\theta^{j'} =\theta^{j} $, $z_t^j$ can be added in $\mathbb{S}_t^j$ without shifting. Then by Theorem \ref{thm_noninc}, $J^{j'}(x_s)\leq J_{0|0}^{j,j'}(x_s)=J^{j}(x_s)$, which means the closed-loop iteration cost does not increase.

\begin{corollary}

Under Assumption~\ref{ass:cost_eco}, considering that the system \ref{dynamics} is controlled by the proposed periodic LMPC \eqref{prob_LMPC} and \eqref{clp_input}, if at the $j$th iteration, it achieves a steady-state solution $\{\mathbf{z}^{j,ss},$ $\mathbf{v}^{j,ss}\}$ with respect to $\theta^j$, then $\{\mathbf{z}^{j,ss},$ $\mathbf{v}^{j,ss}\}$ is the optimal solution of~\eqref{prob_general_robust}.

\begin{proof}
The proof follows a similar procedure to that in~\cite[Theorem 1]{rosolia2020optimality} as~\eqref{prob_general_robust} is strictly convex. 
\end{proof}
\end{corollary}

\section{Practical Issues} \label{sect:5}
In practice, the initial state of each iteration is not necessarily the same, i.e. $\exists\; i<j, z_0^i \ne z_0^j$. For example, even if the building controller is idle in the evening and the system state converges to a steady state due to the dissipative nature of a building, the resulting steady state also varies due to external temperature.

One solution is to involve the initial state deviation as part of the disturbance function $w_t$. By defining a nominal initial state $x_{s,n}$ and the deviation between it and initial state at $j$th iteration $w_s^j = z_0^j - x_{s,n}$, an extension of the disturbance function is
\begin{align}
    w_{\theta^j,t}^j(w_s^j)  =\left\{
        \begin{array}{lr}
        w_s^j, & \; \, t=-1 \\
        w_{\theta^j,t}^j, & o.w.
        \end{array}
\right.
\end{align}

This has an influence on the shifting procedure (\ref{sys_nom:err_dyn}) starting from time $0$,
\begin{align}
    e_{k+1|0}^{i,j} &= (A_k+B_k K_k)e_{k|0}^j + C_k(w_{\theta^j,k}^j(w_s^j)-w_{\theta^i,k}^i(w_s^i)) \notag\\
    v_{k|0}^{i,j} &= v_k^i+ K_k e_{k|0}^{i,j} \notag\\
    z_{k|0}^{i,j} &= z_k^{i} + e_{k|0}^{i,j}, \forall k \in\mathbb{N}_0^T\notag\\
\end{align}
and $e_{0|0}^{i,j} = w_s^j-w_s^i$, and the feasible disturbance set $\mathbb{W}_0^{i}$ for $\{z_0,\theta\}$ at time $0$ is recomputed by the above error dynamics.

Similarly, if the dynamics of system (\ref{dynamics}) vary from iteration to iteration. Define the nominal dynamics matrices as $\overline{A}_t, \overline{B}_t$, and the dynamics deviation
$dA_t^j = A_t^j - \overline{A}_t$, and further assume that $K_t$ stabilizes all possible $A_t^j+B_t^j$. A new shifting procedure starting from time $t$ is,
\begin{align}
    e_{k+1|t}^{i,j} =& (A_k^j+B_k^j K_k)e_{k|t}^j + C_k(w_{\theta^j,k}^j
    -w_{\theta^i,k}^i)\notag\\ &+ (dA_k^j-dA_k^i)*z_k^i+(dB_k^j-dB_k^i)*v_k^i \notag\\
    v_{k|t}^{i,j} =& v_k^i+ K_k e_{k|t}^{i,j} \notag\\
    z_{k|t}^{i,j} = &z_k^{i} + e_{k|t}^{i,j}, \forall k \in\mathbb{N}_t^T\notag\\
\end{align}
and $e_{t|t}^{i,j} = 0$, and the new feasible disturbance set $\mathbb{W}_t^{i}$ for $\{A_t,B_t,\theta\}$ is computed based on that.

\section{Simulation and results}\label{sect:6}
In this section, the proposed LMPC is tested on a spring-mass system and a single zone building model. The first case involves periodic dynamics, periodic constraints, periodic stage costs and a sinusoidal disturbance.  The latter case considers a periodic tracking task, where scheduled comfort conditions on temperatures and three different correlated real-world disturbances decomposition are considered. 
% In this section, the proposed LMPC is tested on a single zone building model, where we consider a periodic tracking task, where scheduled comfort conditions on temperatures and three different correlated real-world disturbance decompositions are considered. 

\subsection{Spring-mass system}
% The simulation result can be found in the extended version of this paper on \href{https://arxiv.org/abs/2011.13781}{ArXiv.org}.
We test the proposed robust LMPC on a spring-mass system $x_{t+1} = A_tx_t +B_tu_t + w_t$, which executes a periodic task of length $T=50$ and the corresponding time-varying dynamics is captured by:
\begin{align*}
    A_t = \begin{bmatrix}
            1 & 0.1 \\
            0.1(1-sin(2\pi t/T)) & 1
            \end{bmatrix},
    B_t = \begin{bmatrix}
            0\\
            0.1
            \end{bmatrix}   
\end{align*}
The disturbance is governed by a biased sinusoidal behavior:
\begin{align*}
    w_t = a_0 + a_1 sin(2\pi t/T)\;,
\end{align*}
where the parameters $a_0$ and $a_1$ are bounded with $a_0 \in \begin{bmatrix}
            -0.1, & 0.1 \\
            -0.1, & 0.1
            \end{bmatrix}$\;,\;
$a_1 \in \begin{bmatrix}
            -0.1, & 0.1 \\
            -0.1, & 0.1
            \end{bmatrix}$. The system is subject to a fixed input constraints, periodic state constraints and optimized over periodic stage cost:
\begin{align*}
    u \in [-10,10],
    &\left\{
        \begin{array}{lr}
        \,[-1,-3]^T\leq x \leq [4,3]^T, & t < T/2 \\
        \,[-4,-3]^T\leq x \leq [1,3]^T, & t \geq T/2
        \end{array}
        \right. ,\notag\\
     l_t(x_t,u_t) = &\left\{
        \begin{array}{lr}
        ||x_{1,t}-x_{1,ref}||_2^2 + ||u_t||_2^2, & t < T/2 \\
        ||x_{1,t}+x_{1,ref}||_2^2 + ||u_t||_2^2, & t \geq T/2
        \end{array}
        \right.,\notag       
\end{align*}
where $x_t = [x_{1,t},x_{2,t}]^T$ and $x_{1,ref} = 2$ so that the periodic stage cost is induced by a switching set-point.

The experiment is carried out with $x_s = [3;0]^T$ and a prediction horizon $N=4$. At each iteration, $a_1$ and $a_2$ are uniformly sampled from its domain. The feedback gain $K_t$ in~\eqref{sys_nom:err_dyn} is chosen as the LQR gain computed with $Q=I$ and $R=1$. 

In \textbf{Figure} \ref{fig1_cost}, the optimal cost, which corresponds to the solution of the problem~\eqref{prob_general_robust}, is time varying because $\{a_0, a_1\}$ change values among iterations. We notice that the cost difference between $J_{{LMPC}}$ and $J^{\ast}$ tends to diminish and the LMPC solution converges to the optimal solution. \textbf{Figure} \ref{fig1_shift} shows that the closed-loop cumulative cost is upper bounded by any shifted cumulative cost of historical iterations, as promised by Theorem \ref{thm_noninc}. 
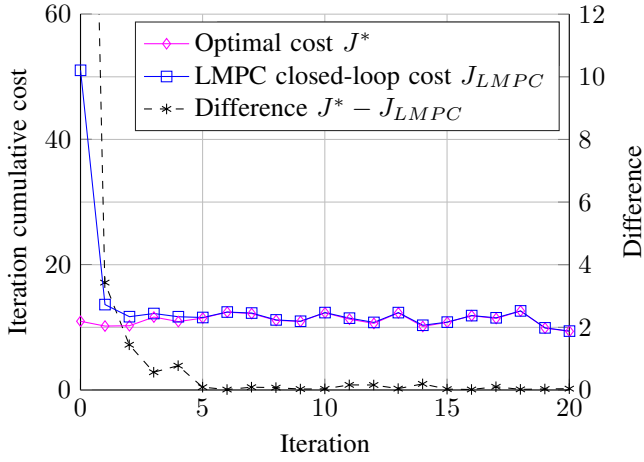
\begin{figure}[htbp]
\centering
% This file was created by matlab2tikz.
%
%The latest updates can be retrieved from
%  http://www.mathworks.com/matlabcentral/fileexchange/22022-matlab2tikz-matlab2tikz
%where you can also make suggestions and rate matlab2tikz.
%
\definecolor{mycolor1}{rgb}{1.00000,0.00000,1.00000}%
\definecolor{mycolor2}{rgb}{0.85000,0.32500,0.09800}%

\begin{tikzpicture}

\begin{axis}[%
width = 6.5cm,
height = 5cm,
scale only axis,
grid=major,
axis x line*=bottom,
axis y line*=right,
xmin=0,
xmax=20,
xlabel={Iteration},
ymin=0,
ymax=12,
ylabel={Difference},
legend style={legend cell align=left, align=left, draw=white!15!black}
]

\addplot [color=black, dashed, mark=asterisk, mark options={solid, black}]
  table[row sep=crcr]{%
0	40.071203185366\\
1	3.42855566367138\\
2	1.44347066200014\\
3	0.566726148132691\\
4	0.775346598294586\\
5	0.0901441705544119\\
6	0.0134500901338939\\
7	0.0894241156550724\\
8	0.0747480596155476\\
9	0.0295723334919433\\
10	0.0359314695032449\\
11	0.162011485726005\\
12	0.159083615887011\\
13	0.0433669476646781\\
14	0.194092137249578\\
15	0.0272411934640058\\
16	0.0156179669923926\\
17	0.102454921889631\\
18	0.020014761250108\\
19	0.0315721450581048\\
20	0.0419530874606906\\
};
\addlegendentry{Difference $J^{\ast}-J_{LMPC}$}

\end{axis}

\begin{axis}[%
width = 6.5cm,
height = 5cm,
scale only axis,
axis x line=none,
axis y line*=left,
xmin=0,
xmax=20,
ymin=0,
ymax=60,
ylabel={Iteration cumulative cost},
legend style={legend cell align=left, align=left, draw=white!15!black}
]

\addplot [color=mycolor1, mark=diamond, mark options={solid, mycolor1}]
  table[row sep=crcr]{%
0	10.9355806550508\\
1	10.1708757476224\\
2	10.2226521706689\\
3	11.6072901894528\\
4	10.8848321237311\\
5	11.4491608474268\\
6	12.3973930553465\\
7	12.1605654058721\\
8	11.0779706284117\\
9	10.9110593883071\\
10	12.2884275506133\\
11	11.2719341106073\\
12	10.5952415865758\\
13	12.2767460146195\\
14	10.1053987720775\\
15	10.7727381030392\\
16	11.826057482139\\
17	11.3873084878357\\
18	12.5716326884463\\
19	9.85264918682547\\
20	9.32683496885957\\
};
\addlegendentry{Optimal cost $J^{\ast}$}

\addplot [color=blue, mark=square, mark options={solid, blue}]
  table[row sep=crcr]{%
0	51.0067838404167\\
1	13.5994314112938\\
2	11.666122832669\\
3	12.1740163375855\\
4	11.6601787220257\\
5	11.5393050179813\\
6	12.4108431454804\\
7	12.2499895215272\\
8	11.1527186880273\\
9	10.940631721799\\
10	12.3243590201165\\
11	11.4339455963333\\
12	10.7543252024628\\
13	12.3201129622842\\
14	10.299490909327\\
15	10.7999792965032\\
16	11.8416754491314\\
17	11.4897634097253\\
18	12.5916474496964\\
19	9.88422133188357\\
20	9.36878805632026\\
};
\addlegendentry{LMPC closed-loop cost $J_{LMPC}$}

\addplot [color=black, dashed, mark=asterisk, mark options={solid, black}]
  table[row sep=crcr]{%
0	2000\\
};
\addlegendentry{Difference $J^{\ast}-J_{LMPC}$}
\end{axis}
\end{tikzpicture}%
\caption{Cumulative cost of each iteration}
\label{fig1_cost}
\end{figure}

\begin{figure}[htbp]
\centering
% This file was created by matlab2tikz.
%
%The latest updates can be retrieved from
%  http://www.mathworks.com/matlabcentral/fileexchange/22022-matlab2tikz-matlab2tikz
%where you can also make suggestions and rate matlab2tikz.
%
\definecolor{mycolor1}{rgb}{1.00000,0.00000,1.00000}%
\definecolor{mycolor2}{rgb}{0.85000,0.32500,0.09800}%

\begin{tikzpicture}

\begin{axis}[%
width = 6.5cm,
height = 5cm,
scale only axis,
grid=major,
xmin=0,
xmax=14,
xlabel={Iteration},
ymin=0,
ymax=70,
ylabel={Shifted cumulative cost},
axis x line*=bottom,
axis y line*=left,
legend style={font=\scriptsize}
]

\addplot [color=green, mark=o, mark options={solid, green}]
  table[row sep=crcr]{%
0	66.1281422094552\\
1	23.2818745455361\\
2	12.3824228876033\\
3	21.4238481995835\\
4	38.3314135632827\\
};
\addlegendentry{Shifted cumulative cost: 5th iteration}

\addplot [color=mycolor1, mark=diamond, mark options={solid, mycolor1}]
  table[row sep=crcr]{%
0	56.5782343961673\\
1	17.2305406893372\\
2	12.9037327123386\\
3	15.2988071257825\\
4	26.7022344780638\\
5	13.8054918510988\\
6	15.4505844964728\\
7	16.8988097936823\\
8	18.9062140153362\\
9	14.761439583391\\
};
\addlegendentry{Shifted cumulative cost: 10th iteration}

\addplot [color=blue, mark=square, mark options={solid, blue}]
  table[row sep=crcr]{%
0	49.718826550935\\
1	13.9347584238187\\
2	12.0975019797278\\
3	13.3316281853904\\
4	24.2994780163024\\
5	13.3350256854474\\
6	13.5799201285473\\
7	13.6745661163116\\
8	15.5004098017435\\
9	12.6810817093658\\
10	11.1520177822925\\
11	13.8012017594572\\
12	12.2969669637185\\
13	11.97362431494\\
14	17.2149475822545\\
};
\addlegendentry{Shifted cumulative cost: 15th iteration}

\addplot [color=green, dashed]
  table[row sep=crcr]{%
0	11.5393050179813\\
1	11.5393050179813\\
2	11.5393050179813\\
3	11.5393050179813\\
4	11.5393050179813\\
};
\addlegendentry{Closed-loop cumulative cost: 5th iteration}

\addplot [color=mycolor1, dashed]
  table[row sep=crcr]{%
0	12.3243590201165\\
1	12.3243590201165\\
2	12.3243590201165\\
3	12.3243590201165\\
4	12.3243590201165\\
5	12.3243590201165\\
6	12.3243590201165\\
7	12.3243590201165\\
8	12.3243590201165\\
9	12.3243590201165\\
};
\addlegendentry{Closed-loop cumulative cost: 10th iteration}

\addplot [color=blue, dashed]
  table[row sep=crcr]{%
0	10.7999792965032\\
1	10.7999792965032\\
2	10.7999792965032\\
3	10.7999792965032\\
4	10.7999792965032\\
5	10.7999792965032\\
6	10.7999792965032\\
7	10.7999792965032\\
8	10.7999792965032\\
9	10.7999792965032\\
10	10.7999792965032\\
11	10.7999792965032\\
12	10.7999792965032\\
13	10.7999792965032\\
14	10.7999792965032\\
};
\addlegendentry{Closed-loop cumulative cost: 15th iteration}
\end{axis}

\end{tikzpicture}%
\caption{Comparison of shifted and closed-loop cumulative cost starting from time $0$}
\label{fig1_shift}
\end{figure}
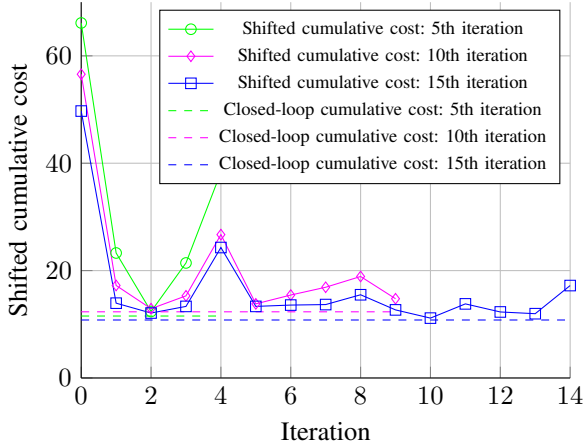

\subsection{A single zone building system}
A small scale linear time invariant building model\cite{oldewurtel2008tractable} with $x_{t+1} = Ax_t + Bu_t + Cw_t$ is considered, where
\begin{align*}
    A &= \begin{bmatrix}
            0.8511 & 0.0541 & 0.0707 \\
            0.1293 & 0.8635 & 0.0055 \\
            0.0989 & 0.0032 & 0.7541 \\
            \end{bmatrix},
    B = \begin{bmatrix}
            0.0035\\
            0.0003\\
            0.0002
            \end{bmatrix} \notag\\    
    C &= 10^{-3}*\begin{bmatrix}
            22.2170 & 1.7912 & 42.2123 \\
            1.5376 & 0.6944 & 2.9214 \\
            103.1813 & 0.1032 & 196.0444 \\
            \end{bmatrix}\;.
\end{align*}
The states $x=[x_1,x_2,x_3]^T$ represent the temperatures of the room, the wall connected with another room, and the wall connected to the outside respectively. The single input is heating and cooling. Suppose the sampling rate of the system is $10$ minutes; a one-day iteration consists of $144$ time steps.

In this test, the disturbances of internal heat-gains, solar-radiation and external temperature are considered, which are denoted by $w = [w_1,w_2,w_3]^T$ accordingly. These disturbances all reveal daily repetitive patterns and can be predicted~\cite{oldewurtel2012use}. For the sake of simplicity, We use the combination of sinusoidal, triangular and square wave functions and white noise to approximate the decomposition of disturbances in~\eqref{disturbance_finite}:
\begin{align*}
    w_{1,t} &=  a_1+a_2sin(2\pi t/T)+w_{r,1,t}\\
    w_{2,t} &= \left\{
        \begin{array}{ll}
        a_3(4t-T)/T+w_{r,2,t}, & T/4\leq t < T/2\\
        a_3(3T-4t)/T+w_{r,2,t}, & T/2\leq t < 3T/4\\
        w_{r,2,t}, & t< T/4 \vee t \geq 3T/4
        \end{array}
        \right.\\
    w_{3,t} &= \left\{
        \begin{array}{ll}
        a_4+a_5+w_{r,3,t}, & T/3\leq t < 3T/4\\
        a_4+w_{r,3,t}, & t<T/3 \vee t \geq3T/4
        \end{array}
        \right.    
\end{align*}
, where the parameters and white noise are bounded by:
\begin{align*}
    &a_1\in [10,14], a_2\in [-6,-2] \\
    &a_3\in [0,16],a_4\in [0,2],a_5\in [6,7] \\
    &w_{r,1} \in [-3,3],w_{r,2} \in [-5,5],w_{r,3} \in [-2,2]
\end{align*}

The room temperature is required to satisfy a comfort constraint during work time and the constraint is relaxed at night. The constraints are modeled as:
\begin{align*}
u \in [-30,30],
\left\{
    \begin{array}{ll}
    \,18\leq x_1 \leq 30, & t < T/3 \vee t \geq 3T/4 \\
    \,22\leq x_1 \leq 26, & T/3 \leq t \leq 3T/4
    \end{array}
    \right.
\end{align*}

The control objective is to regulate the room temperature to a time-varying reference
\begin{align*}
    x_{1,ref,t}=\left\{
    \begin{array}{ll}
    20, & t < T/3 \vee t \geq 3T/4 \\
    24, & T/3 \leq t < 3T/4
    \end{array}
    \right.
\end{align*}
while minimizing the energy cost.
The stage cost is $l_t(x_t,u_t) = ||x_{1,t}-x_{1,ref,t}||_2^2 + ||cp_t u_t||_1$, in which $cp_t$ denotes the electricity price and there are periodic high price and low price periods:
 \begin{align*}
    cp_t=&\left\{
        \begin{array}{ll}
        \,1, & \, t < 5T/12 \vee t \geq 2T/3 \\
        \,2, & \, 5T/12 \leq T < 2T/3
        \end{array}
        \right.
\end{align*}

The experiment is carried out with an initial state $x_s = [19;19;15]^T$ and prediction horizon $N=16$. The feedback gain $K$ in~\eqref{sys_nom:err_dyn} and~\eqref{sys:error} is computed by the optimal LQR gain choosing parameters $Q=10I$ and $R=1$. The constraints are tightened by robust positive invariant $\varepsilon$ in~\eqref{sys:tighten}, which is computed by an approximation method in~\cite{rakovic2005invariant}. The noise parameters $\{a_1, a_2, a_3, a_4, a_5\}$ and  the white noise $w_{r,1},w_{r,2}, w_{r,3}$ are uniformly sampled from their domain. 

% \textbf{Figure} \ref{fig2_shift} shows that the closed-loop cumulative cost is upper bounded by any shifted cumulative cost of history iterations. 

In \textbf{Figure} \ref{fig2_cost}, the cumulative cost of LMPC converges to the optimal cumulative cost. In particular, the optimal cost refers to the optimal solution of problem~\eqref{prob_general_robust}. Note the cost difference between $J_{\text{LMPC}}$ and $J^{\ast}$ does not decrease monotonically due to the shifted trajectories. However,  \textbf{Figure} \ref{fig2_shift} shows that the closed-loop cumulative cost from $t=0$ is upper bounded by any shifted cumulative cost from $t=0\,$, guaranteed by Theorem \ref{thm_noninc}.
% And the final convergent state trajectory is shown in \textbf{Figure}. \ref{fig2_sys_20}. 
% The trajectory evolution along the iterations can be found in the \href{https://arxiv.org/abs/2011.13781}{\underline{extended version}}.
In \textbf{Figure} \ref{fig2_sys_1}, \ref{fig2_sys_3}, \ref{fig2_sys_20}, we compare the first state $x_1$ among the $0$th shifted trajectory, trajectory by LMPC and the optimal trajectory at iteration $1,3,20$, which shows the convergence toward the optimal solution $J^*$. And the final convergent state trajectory is shown in \textbf{Figure}. \ref{fig2_sys_20}.

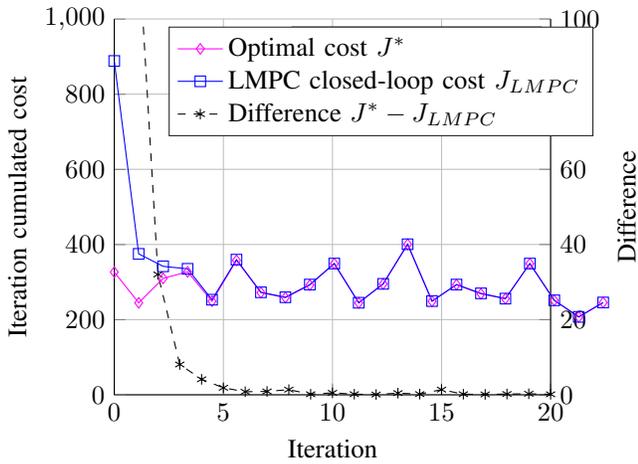
\begin{figure}[htbp]
\centering
% This file was created by matlab2tikz.
%
%The latest updates can be retrieved from
%  http://www.mathworks.com/matlabcentral/fileexchange/22022-matlab2tikz-matlab2tikz
%where you can also make suggestions and rate matlab2tikz.
%
\definecolor{mycolor1}{rgb}{1.00000,0.00000,1.00000}%
\definecolor{mycolor2}{rgb}{0.85000,0.32500,0.09800}%

\begin{tikzpicture}

\begin{axis}[%
width = 5.8cm,
height = 5cm,
scale only axis,
grid=major,
axis x line*=bottom,
axis y line*=right,
xmin=0,
xmax=20,
xlabel={Iteration},
ymin=0,
ymax=100,
ylabel={Difference},
legend style={legend cell align=left, align=left, draw=white!15!black,font=\tiny}
]

\addplot [color=black, dashed, mark=asterisk, mark options={solid, black}]
  table[row sep=crcr]{%
0	561.783780484232\\
1	130.260096247504\\
2	32.0182596837835\\
3	8.06413143984952\\
4	4.11736739837394\\
5	1.86449683304932\\
6	0.8123601790424\\
7	0.888805919012952\\
8	1.36746382483147\\
9	0.148952158844395\\
10	0.51651036636747\\
11	0.164886202849573\\
12	0.0725344210528078\\
13	0.442729528427037\\
14	0.184699042316709\\
15	1.37400181820061\\
16	0.159421685900327\\
17	0.103543471179364\\
18	0.206324989347053\\
19	0.269811258052385\\
20	0.131239875716517\\
};
\addlegendentry{Difference $J^{\ast}-J_{LMPC}$}

\end{axis}

\begin{axis}[%
width = 6.5cm,
height = 5cm,
scale only axis,
axis x line=none,
axis y line*=left,
xmin=0,
xmax=20,
ymin=0,
ymax=1000,
ylabel={Iteration cumulated cost},
legend style={legend cell align=left, align=left, draw=white!15!black}
]

\addplot [color=mycolor1, mark=diamond, mark options={solid, mycolor1}]
  table[row sep=crcr]{%
0	326.802932352713\\
1	245.030355712212\\
2	310.079151055958\\
3	327.430411549252\\
4	249.250174930154\\
5	358.109656020727\\
6	272.202321672855\\
7	258.939323977525\\
8	292.265586912023\\
9	349.371080571215\\
10	244.733167413117\\
11	295.337342149277\\
12	400.846957992817\\
13	249.184374951633\\
14	293.409484464002\\
15	268.822444497747\\
16	256.068749277233\\
17	349.400485187715\\
18	251.50663670727\\
19	207.410368506045\\
20	245.817631978199\\
};
\addlegendentry{Optimal cost $J^{\ast}$}

\addplot [color=blue, mark=square, mark options={solid, blue}]
  table[row sep=crcr]{%
0	888.586712836944\\
1	375.290451959716\\
2	342.097410739741\\
3	335.494542989102\\
4	253.367542328528\\
5	359.974152853777\\
6	273.014681851897\\
7	259.828129896538\\
8	293.633050736855\\
9	349.520032730059\\
10	245.249677779484\\
11	295.502228352126\\
12	400.91949241387\\
13	249.62710448006\\
14	293.594183506318\\
15	270.196446315948\\
16	256.228170963134\\
17	349.504028658895\\
18	251.712961696617\\
19	207.680179764098\\
20	245.948871853916\\
};
\addlegendentry{LMPC closed-loop cost $J_{LMPC}$}

\addplot [color=black, dashed, mark=asterisk, mark options={solid, black}]
  table[row sep=crcr]{%
0	2000\\
};
\addlegendentry{Difference $J^{\ast}-J_{LMPC}$}
\end{axis}
\end{tikzpicture}%
\caption{Cumulative cost of each iteration}
\label{fig2_cost}
\end{figure}

\begin{figure}[htbp]
\centering
% This file was created by matlab2tikz.
%
%The latest updates can be retrieved from
%  http://www.mathworks.com/matlabcentral/fileexchange/22022-matlab2tikz-matlab2tikz
%where you can also make suggestions and rate matlab2tikz.
%
\definecolor{mycolor1}{rgb}{1.00000,0.00000,1.00000}%
\definecolor{mycolor2}{rgb}{0.85000,0.32500,0.09800}%

\begin{tikzpicture}

\begin{axis}[%
width = 6.5cm,
height = 5cm,
scale only axis,
grid=major,
xmin=0,
xmax=14,
xlabel={Iteration},
ymin=200,
ymax=1000,
ylabel={Shifted cumulative cost},
axis x line*=bottom,
axis y line*=left,
legend style={font=\scriptsize}
]

\addplot [color=green, mark=o, mark options={solid, green}]
  table[row sep=crcr]{%
0	924.983964121269\\
1	484.918752530683\\
2	391.857301970996\\
3	367.897099867655\\
4	367.188064672156\\
};
\addlegendentry{Shifted cumulative cost: 5th iteration}

\addplot [color=mycolor1, mark=diamond, mark options={solid, mycolor1}]
  table[row sep=crcr]{%
0	823.330177266517\\
1	375.398760326215\\
2	293.206878785662\\
3	267.880862498387\\
4	251.015082351574\\
5	266.255487973472\\
6	251.411219947396\\
7	249.565221154253\\
8	256.479891250746\\
9	260.896086626662\\
};
\addlegendentry{Shifted cumulative cost: 10th iteration}

\addplot [color=blue, mark=square, mark options={solid, blue}]
  table[row sep=crcr]{%
0	881.606880159329\\
1	411.918191409073\\
2	323.59822724094\\
3	296.517058964208\\
4	281.109513911204\\
5	293.262150326431\\
6	280.030055052483\\
7	276.543513609491\\
8	282.452426000402\\
9	288.483404713102\\
10	275.788120413057\\
11	282.309187192427\\
12	291.468353575573\\
13	276.041287020983\\
14	281.820265670881\\
};
\addlegendentry{Shifted cumulative cost: 15th iteration}

\addplot [color=green, dashed]
  table[row sep=crcr]{%
0	359.974152853777\\
1	359.974152853777\\
2	359.974152853777\\
3	359.974152853777\\
4	359.974152853777\\
};
\addlegendentry{Closed-loop cumulative cost: 5th iteration}

\addplot [color=mycolor1, dashed]
  table[row sep=crcr]{%
0	245.249677779484\\
1	245.249677779484\\
2	245.249677779484\\
3	245.249677779484\\
4	245.249677779484\\
5	245.249677779484\\
6	245.249677779484\\
7	245.249677779484\\
8	245.249677779484\\
9	245.249677779484\\
};
\addlegendentry{Closed-loop cumulative cost: 10th iteration}

\addplot [color=blue, dashed]
  table[row sep=crcr]{%
0	270.196446315948\\
1	270.196446315948\\
2	270.196446315948\\
3	270.196446315948\\
4	270.196446315948\\
5	270.196446315948\\
6	270.196446315948\\
7	270.196446315948\\
8	270.196446315948\\
9	270.196446315948\\
10	270.196446315948\\
11	270.196446315948\\
12	270.196446315948\\
13	270.196446315948\\
14	270.196446315948\\
};
\addlegendentry{Closed-loop cumulative cost: 15th iteration}
\end{axis}

\end{tikzpicture}%
\caption{Comparison of shifted and closed-loop cumulative cost starting from time $0$}
\label{fig2_shift}
\end{figure}
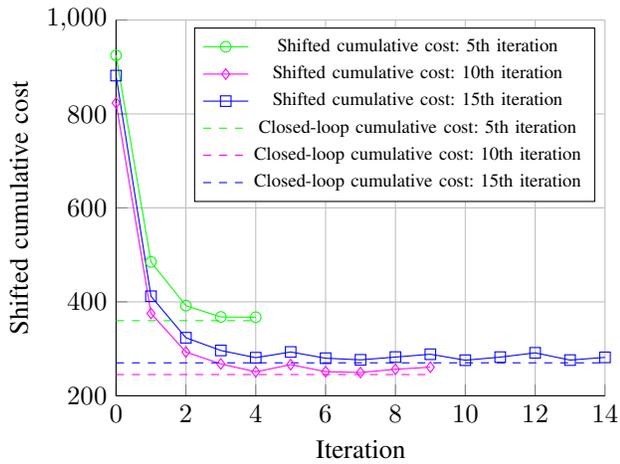

\begin{figure}[htbp]
\centering
\input{sys2_1.tex}
\caption{Building system:$x_1$ at iteration 1}
\label{fig2_sys_1}
\end{figure}

\begin{figure}[htbp]
\centering
\input{sys2_3.tex}
\caption{Building system:$x_1$ at iteration 3}
\label{fig2_sys_3}
\end{figure}

\begin{figure}[htbp]
\centering
\input{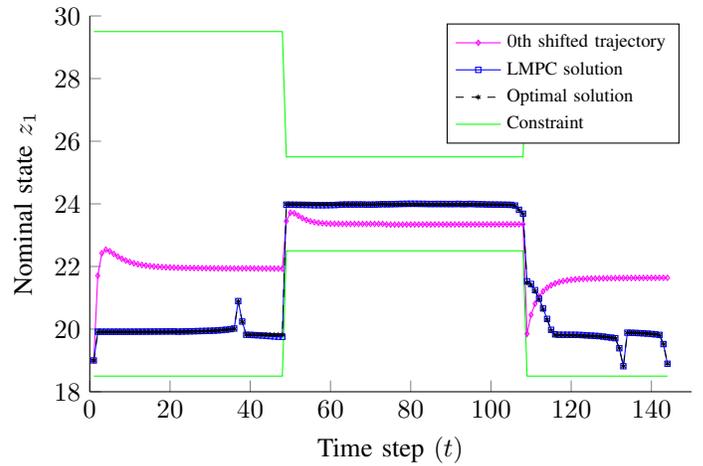}
\caption{Building system:$x_1$ at iteration 20}
\label{fig2_sys_20}
\end{figure}

\section{Conclusion}
We presented a novel, less conservative, robust LMPC scheme for  periodically correlated process noise in building control. The framework is specified for time-varying iterative tasks with periodicity in the system dynamics, stage cost and constraints. Feasibility and performance convergence are verified by a single zone building system.

\bibliographystyle{./bibliography/IEEEtran}
\bibliography{ref.bib}

\section{Appendix}
\subsection{Robust and stochastic LMPC} \label{app:robust}
The long-horizon optimal problem~\eqref{prob_general} is difficult to solve because the stochastic $w_{r,t}^j$ leads to a stochastic optimization objective and it optimizes over all possible control policy.   
A possible approach to deal with the problem is the tube method with a nominal optimization objective~\cite{kouvaritakis2016model}. Denote by $z_t^j$ the nominal state, by $e_t^j = x_t^j - z_t^j$ the error state, by $v_t^j$ the nominal input, and by $Ke(k)$ the tube controller, where $K$ stabilize all different $A_t+B_tK$. Then the tube controller is defined as
\begin{align} \label{sys:error}
    z_{t+1}^j &= A_t z_t^j + B_t v_t^j + C_t w_{\theta^j,t}^j, \notag\\
    e_{t+1}^j &= (A_t+B_t K)e_t^j +  C_t w_{r,t}^j \notag\\
    u_t^j &= K e_t^j + v_t^j
\end{align}
and $z_0^j = x_s$. 
Compute the \textit{Robust Positive Invariant set} $\varepsilon$ of $e_t$ with dynamics~\eqref{sys:error}. Then a constraint tightening is applied on the nominal system:
\begin{align} \label{sys:tighten}
  F_t z_t + G_t v_t \leq f_t - (F_t+G_t K)e_t, \forall e_t \in \varepsilon.
\end{align}
Thus, optimize the problem over the nominal stage cost $l_t(z_t,v_t)$ with the constraints~\eqref{sys:tighten}, a robust problem in~\eqref{prob_general_robust} is derived.

\end{document}